\documentclass[11pt]{article}
\usepackage[utf8]{inputenc}
\usepackage{amsmath,amsfonts,amssymb}
\usepackage{geometry}
\usepackage{graphicx}
\usepackage{natbib}
\usepackage{hyperref}
\usepackage{booktabs}
\usepackage{tabularx}
\usepackage{amsthm}

\newtheorem{remark}{Remark}
\newtheorem{proposition}{Proposition}
\newtheorem{corollary}{Corollary}
\title{Geometric Dynamics of Consumer Credit Cycles: A Multivector-based Linear-Attention Framework for Explanatory Economic Analysis}
\author{Agus Sudjianto$^{1,2}$ and Sandi Setiawan$^3$\\
$^1$H2O.ai, \texttt{agus.sudjianto@h2o.ai} \\
$^2$Center for Trustworthy AI Through Model Risk Management, \\
University of North Carolina Charlotte\\
$^3$Life Member, Clare Hall, University of Cambridge, \texttt{sandi.setiawan@cantab.net}}
\date{October 2025}

\begin{document}

\maketitle

\begin{abstract}
Consumer credit cycles exhibit complex dynamics where identical economic shocks can produce vastly different outcomes—from contained recessions to devastating financial crises. The key lies in understanding whether variable relationships remain linear or transition into dangerous feedback spirals that amplify initial disturbances. Traditional econometric approaches, relying on scalar correlations, cannot distinguish between these fundamentally different interaction regimes that drive credit cycle severity.

This study introduces geometric algebra to decompose credit system relationships into their projective (correlation-like) and rotational (feedback-spiral) components. We represent economic states as multivectors in Clifford algebra, where bivector elements capture the rotational coupling between unemployment, consumption, savings, and credit utilization. This mathematical framework reveals interaction patterns invisible to conventional analysis: when unemployment and credit contraction enter simultaneous feedback loops, their geometric relationship shifts from simple correlation to dangerous rotational dynamics that characterize systemic crises.

We implement this geometric framework using linear attention mechanisms that identify historical periods with similar multivector signatures to current conditions. This machine learning approach enables dynamic pattern recognition, allowing the model to adapt its understanding of credit relationships based on geometric similarity to past episodes rather than assuming stable parameters across all economic regimes.

Applied to U.S. quarterly data (1980-2024), the framework reveals three distinct credit cycle regimes with different geometric signatures. Normal periods show vector-dominated patterns with predictable lead-lag relationships. The 1990-91 recession exhibited contained stress with moderate geometric coupling, while the 2008 financial crisis demonstrated dangerous bivector dominance reflecting unemployment-credit feedback spirals—despite similar correlation coefficients. Most significantly, current conditions (2022-2024) display geometric patterns resembling manageable cyclical episodes rather than systemic crisis, with credit dynamics remaining largely additive rather than multiplicative. The framework thus provides early warning capabilities by monitoring the geometric transition that precedes credit cycle amplification.

\end{abstract}

\section{Introduction}

Understanding the dynamics of consumer credit cycles requires analyzing a complex web of interconnected economic relationships. When unemployment rises, consumers typically reduce spending and increase precautionary savings, while simultaneously facing greater difficulty servicing existing debt obligations. This creates pressure on revolving credit balances as households may increase borrowing to maintain consumption levels, ultimately leading to higher default rates. However, the timing, magnitude, and interaction patterns of these relationships vary dramatically across different economic cycles, creating fundamentally different crisis mechanisms that correlation-based analysis cannot distinguish.

Consider the 2008 financial crisis versus the 1990-91 recession. Both periods show high correlations between unemployment and credit defaults ($\rho \approx 0.8$), yet the underlying mechanisms differ fundamentally. In 2008, unemployment and credit contraction formed a reinforcing spiral---each variable amplifying the other in a feedback loop that created the conditions for systemic crisis. As unemployment rose, consumer spending collapsed, leading banks to tighten credit standards. This credit contraction further reduced economic activity, causing additional unemployment, which intensified the spending decline and credit tightening in an accelerating cycle. Meanwhile, desperate consumers drew down savings and maxed out existing credit lines before defaulting, creating a simultaneous collapse in the savings rate and explosion in charge-offs.

In contrast, during the 1990-91 recession, the relationship was more sequential: unemployment rose first due to structural economic adjustments, then defaults followed with a predictable lag pattern as affected households exhausted their financial buffers. The savings rate actually \emph{increased} during this period as unaffected consumers adopted precautionary behavior, and credit tightening was more measured and targeted rather than systematic. The absence of reinforcing feedback loops contained the crisis to a conventional business cycle downturn.

The COVID-19 crisis revealed yet another pattern entirely. Unprecedented unemployment coincided with \emph{rising} savings rates due to fiscal transfers and constrained spending opportunities, while defaults remained surprisingly low due to policy backstops including forbearance programs and direct payments. Here, traditional unemployment-to-default transmission channels were effectively severed by policy intervention, creating a geometric signature unlike any historical precedent.

These examples illustrate a fundamental challenge in economic analysis: the same variables can exhibit similar correlation patterns while operating through entirely different causal mechanisms. Standard econometric models, which estimate global parameters across all time periods, struggle to distinguish between these mechanistically different crises. Vector Autoregressive (VAR) models and traditional regression approaches provide averaged relationships that obscure the time-varying interaction patterns crucial for understanding different crisis dynamics. While these models can identify that unemployment and credit defaults are correlated, they cannot reveal whether this correlation stems from simple lead-lag relationships, complex feedback spirals, or policy-disrupted transmission mechanisms.

\subsection{The Geometric Nature of Economic Relationships}

The interactions between consumer spending, savings, unemployment, revolving credit, and defaults exhibit geometric properties that scalar correlation analysis cannot capture. When unemployment and consumer spending move in opposite directions with a simple lag structure, their relationship is primarily \emph{projective}---one variable moves predictably in response to the other with a measurable delay. However, when these variables enter reinforcing feedback relationships, their interaction becomes \emph{rotational}---each influences the other simultaneously, creating circular causality that amplifies initial shocks.

Consider the dynamics during a typical crisis progression. Initially, an external shock (such as an oil price spike or financial market disruption) increases unemployment. In the projective regime, this leads to predictable reductions in consumer spending and increases in precautionary savings, with modest increases in credit utilization as some households smooth consumption. Default rates rise gradually as the most vulnerable borrowers exhaust their resources.

The critical transition occurs when these linear relationships become circular. As defaults mount, banks tighten credit standards not just for new borrowers, but for existing customers through line reductions and rate increases. This credit tightening forces even employed consumers to reduce spending, which reduces business revenues and increases unemployment. The newly unemployed then face both lost income \emph{and} reduced credit access, accelerating their path to default. Meanwhile, employed consumers witness this deterioration and increase precautionary savings, further reducing spending and amplifying the economic contraction.

This rotational dynamic creates what we term \emph{geometric amplification}---small initial shocks get magnified through circular feedback processes that correlational analysis interprets simply as ``high correlation periods.'' The geometric algebra framework we propose can distinguish between these fundamentally different interaction regimes by decomposing variable relationships into their projective (inner product) and rotational (outer product) components.

\subsection{Traditional Approaches and Their Limitations}

A fundamental limitation of traditional econometric approaches is their sequential treatment of magnitude and phase relationships. Researchers typically first estimate correlations or covariance structures between variables, then separately analyze temporal patterns through Granger causality tests, impulse response functions, or cross-correlation analysis at different lags. This two-step process assumes that correlation structures remain stable while phase relationships are analyzed---an assumption that breaks down precisely during crisis periods when both magnitude and timing relationships evolve simultaneously.

During the 2008 financial crisis, for example, the correlation between unemployment and credit defaults remained high throughout the period, but the underlying interaction pattern shifted from a simple lead-lag relationship (unemployment leads defaults) to a simultaneous feedback spiral where each variable amplified the other. Traditional methods would require separate models to capture these different regimes, making it difficult to understand the transition dynamics that are often most critical for policy intervention.

Moreover, the relationships between consumer spending, savings, unemployment, credit balances, and defaults exhibit time-varying interaction effects that standard VAR models struggle to capture. The impact of unemployment on consumer spending depends critically on the state of credit markets---when credit is readily available, unemployment may have modest effects on spending as households can smooth consumption through borrowing. When credit markets are stressed, the same level of unemployment can cause dramatic spending reductions as households face binding credit constraints.

Similarly, the relationship between savings and spending is not simply negative (as standard permanent income theory would suggest) but depends on the prevailing uncertainty regime. During normal times, higher savings may indicate confidence-driven consumption deferral with minimal impact on current spending. During crisis periods, higher savings may reflect panic-driven consumption cuts that amplify economic contractions. These regime-dependent interaction effects cannot be captured by models that assume stable, linear relationships between variables.

\subsection{Our Approach and Key Contributions}

This study introduces a framework that captures the \emph{geometric structure} of variable interactions, enabling analysis of how the relationships between consumer spending, savings, unemployment, revolving credit, and defaults evolve across different cycles. Our approach moves beyond scalar correlations to examine the rotational dynamics between economic variables---the feedback patterns, phase relationships, and coupling mechanisms that distinguish one crisis from another.

Our framework's primary objective is not to forecast the future but to uncover the historical geometric relationships that govern the evolution of a system through pattern recognition and historical precedent identification. By representing the five-variable system as evolving multivectors in geometric algebra space, we can distinguish between periods where relationships are primarily additive (vector-dominated, corresponding to normal business cycles) versus periods where circular feedback effects dominate (bivector-dominated, corresponding to crisis amplification).

The linear attention mechanism serves as an analytical tool, revealing which past economic configurations are most geometrically similar to current conditions. This enables the model to adapt its interpretation of variable relationships based on historical precedent---recognizing, for instance, that current unemployment may have different implications depending on whether the overall economic configuration more closely resembles the contained stress of 1990-91 or the systemic feedback loops of 2008.

Our key contributions are threefold:

\textbf{1. Geometric Decomposition of Economic Relationships:} We use Geometric Algebra to decompose the interactions between consumer spending, savings, unemployment, revolving credit, and defaults into their projective (correlation-like) and rotational (feedback-spiral) components. This mathematical framework provides a principled way to distinguish between additive economic stress and multiplicative crisis amplification, offering insights that pure correlation analysis cannot provide.

\textbf{2. Time-Varying Attention as an Analytical Tool:} The model's linear attention mechanism serves as a window into dynamic economic relationships. By analyzing the learned attention weights and multivector parameters, we can directly observe which historical periods and which specific geometric relationships are most relevant for explaining the current state of the credit system. This transforms pattern recognition from an implicit model behavior into an explicit analytical output.

\textbf{3. Interpretable Parameter Framework:} Unlike black-box machine learning approaches, our framework produces directly interpretable parameters where bivector coefficients have clear economic meaning as measures of feedback coupling strength between variable pairs. This enables economic analysis rather than just prediction, providing insights into the mechanisms driving credit cycle dynamics.

Applied to U.S. quarterly data (1980--2024), the model demonstrates strong historical fit and reveals interpretable findings: bivector parameters capture rotational relationships between unemployment, consumption, savings, and credit balances that vary significantly across different crisis periods. Unlike correlation-based approaches that provide global measures, our framework reveals time-varying geometric signatures that distinguish mechanistically different crises even when their linear correlations appear similar.

The paper is structured as follows: Section 2 outlines related work and positions our contribution within the literature on geometric algebra applications, attention mechanisms, and economic time series analysis. Section 3 presents the theoretical foundations and model architecture, emphasizing the economic intuition behind geometric relationships. Section 4 details the empirical motivation and analytical methodology, demonstrating why standard correlation analysis proves insufficient for understanding credit cycle dynamics. Section 5 establishes theoretical guarantees for the framework's mathematical properties. Section 6 presents the results, focusing on the discovered geometric dynamics across different crisis periods and their economic interpretation. Section 7 discusses the broader implications for economic research and acknowledges the study's limitations.

\section{Related Works}

Our framework synthesizes advances from three distinct research streams: geometric algebra applications in machine learning, linear attention mechanisms for temporal modeling, and interpretable methods for economic time series analysis. Rather than incrementally improving existing approaches, we demonstrate how their principled combination addresses fundamental limitations in each individual domain.

\subsection{Geometric Algebra in Machine Learning}

The mathematical foundations of Geometric Algebra (GA) were established by \citet{hestenes1984}, who demonstrated its unifying power for representing rotations, reflections, and other geometric transformations in physics. The recent "Geometric Deep Learning" movement \citep{bronstein2017} has applied these principles to neural architectures that preserve geometric structure, with notable successes in computer vision and robotics applications requiring equivariance to spatial transformations.

However, existing GA applications focus predominantly on static spatial data where geometric transformations have clear physical interpretations. The extension to temporal economic data presents novel challenges: economic "rotations" represent interaction dynamics rather than spatial orientations, and the relevant symmetries involve economic relationships rather than coordinate systems. Our work addresses this gap by developing economically meaningful interpretations of multivector operations for time series analysis.

Recent work by \citet{ruhe2023geometric} explored GA for dynamical systems using Geometric Clifford Algebra Networks (GCANs), demonstrating applications to fluid dynamics and physical simulations. However, their focus on systems with well-defined conservation laws and physical symmetries differs fundamentally from economic systems, which lack such invariances and require different theoretical foundations for interpreting geometric relationships between variables.

\subsection{Linear Attention for Interpretable Sequence Modeling}

The attention mechanism's development from \citet{bahdanau2014} to the transformer architecture \citep{vaswani2017} revolutionized sequence modeling, but the quadratic computational complexity motivated research into linear variants. \citet{katharopoulos2020transformers} introduced a practical linear attention mechanism using kernel feature maps, showing that transformers could be reformulated as RNNs with linear complexity. \citet{choromanski2020} further developed the PERFORMER architecture using random feature approximations to achieve similar efficiency gains.

Our contribution differs fundamentally from this efficiency-focused literature. While previous work treats linear attention as an approximation to full attention for computational benefits, we demonstrate that its mathematical structure provides direct economic interpretability. The rational form of linear attention naturally captures relative weighting of historical precedents—a concept central to economic reasoning by analogy but absent from standard attention formulations.

The key insight that linear attention coefficients can represent time-varying economic parameters with direct interpretable meaning appears novel to the economic time series literature, bridging the gap between computational efficiency and economic understanding.

\subsection{Traditional Economic Time Series Methods}

Traditional econometric approaches have long dominated economic time series analysis. Vector Autoregressive (VAR) models \citep{sims1980} and ARMA/ARIMA models \citep{box1970} provide highly interpretable frameworks but struggle with complex, time-varying dynamics. Their fundamental limitation lies in the sequential treatment of magnitude and phase relationships: researchers typically estimate correlation structures first, then separately analyze temporal patterns through Granger causality tests or impulse response functions.

This two-step approach assumes stable correlation structures while investigating phase relationships—an assumption that breaks down during crisis periods when both magnitude and timing relationships evolve simultaneously. The econometric literature has developed various approaches to handle time-varying parameters, including threshold models and regime-switching frameworks, but these typically require pre-specification of regime changes or switching mechanisms.

The rise of machine learning in economics has focused primarily on improving forecasting accuracy. Recent surveys of deep learning applications to macroeconomic forecasting \citep{paranhos2024,almosova2023} show that ensemble methods and neural networks often outperform linear models predictively, but they provide limited insight into the underlying economic mechanisms driving the relationships.

\subsection{Machine Learning for Economic Analysis}

The application of machine learning to economic analysis has generated substantial literature, with most studies emphasizing predictive performance over explanatory insight. While neural networks can capture complex non-linearities and long-term dependencies in economic data, they typically operate as "black box" models that lack transparency regarding the economic relationships they have learned.

This creates a fundamental tension in economic research: more sophisticated models often provide better predictions but less economic understanding. Traditional econometric approaches offer clear interpretations of estimated parameters, while modern ML methods sacrifice this interpretability for improved predictive accuracy.

Our work addresses this trade-off by developing a framework that maintains both mathematical sophistication and economic interpretability. By grounding our approach in geometric algebra and using linear attention as an analytical tool rather than just a computational device, we preserve direct economic meaning in the learned parameters.

\subsection{Positioning Our Contribution}

Our work makes three novel contributions that distinguish it from existing approaches:

\textbf{Theoretical Innovation}: We provide the first systematic framework for applying geometric algebra to economic time series analysis, developing new interpretations of multivector operations for temporal relationships. Unlike existing GA applications that focus on spatial transformations, we demonstrate how bivector components can capture economic feedback mechanisms and interaction dynamics.

\textbf{Methodological Synthesis}: We show that linear attention's mathematical structure provides direct economic interpretability when combined with geometric algebra embeddings. This moves beyond both computational efficiency motivations and post-hoc explanation approaches by building interpretability into the model architecture itself.

\textbf{Empirical Insights}: We demonstrate that geometric relationships can distinguish between mechanistically different economic crises that appear similar under traditional correlation analysis. This provides new tools for understanding crisis dynamics rather than just predicting outcomes, addressing a key limitation in the existing literature.

The framework bridges the gap between traditional econometric interpretability and modern ML capabilities, offering a principled approach to understanding complex economic relationships without sacrificing mathematical rigor or computational efficiency.

\section{Methodology}

\subsection{Foundations of Geometric Algebra}

Geometric Algebra (GA), also known as Clifford Algebra, extends vector spaces by defining a consistent geometric product between vectors that unifies concepts of magnitude, direction, and orientation. For two vectors $\mathbf{a}$ and $\mathbf{b}$, the geometric product $\mathbf{ab}$ decomposes into symmetric and antisymmetric components:
\begin{equation}
\mathbf{ab} = \mathbf{a} \cdot \mathbf{b} + \mathbf{a} \wedge \mathbf{b}
\label{eq:geometric_product}
\end{equation}

The inner product $\mathbf{a} \cdot \mathbf{b}$ yields a scalar measuring the projection of one vector onto another, analogous to correlation but preserving directional information. The outer product $\mathbf{a} \wedge \mathbf{b}$ produces a bivector—a directed planar element that captures the oriented area spanned by the two vectors and encodes their rotational relationship.

\subsubsection{Economic Interpretation of Geometric Operations}

The geometric product decomposition in Equation~\ref{eq:geometric_product} provides direct economic insight into variable relationships:

\textbf{Inner Product as Economic Projection}: The scalar $\mathbf{a} \cdot \mathbf{b}$ measures how much one economic variable moves in the direction of another. Unlike simple correlation, this preserves information about the relative magnitudes and directions of movement, enabling distinction between co-movement and contra-movement patterns.

\textbf{Outer Product as Economic Rotation}: The bivector $\mathbf{a} \wedge \mathbf{b}$ captures the dynamic interplay between variables that correlation analysis cannot detect. For economic variables unemployment ($u$) and credit growth ($v$), the bivector $u \wedge v$ encodes:
\begin{itemize}
\item \textbf{Magnitude}: The strength of feedback coupling between the variables
\item \textbf{Orientation}: The direction of rotational dynamics (stabilizing vs. destabilizing)
\item \textbf{Temporal signature}: Whether interactions follow simple lead-lag patterns or complex feedback loops
\end{itemize}

This geometric interpretation enables us to distinguish between three fundamental interaction regimes:
\begin{enumerate}
\item \textbf{Projective dynamics} (small bivector magnitude): Variables exhibit predictable lead-lag relationships characteristic of normal economic cycles
\item \textbf{Rotational dynamics} (large bivector magnitude): Variables enter feedback loops where each influences the other simultaneously, characteristic of crisis amplification
\item \textbf{Mixed dynamics}: Intermediate regimes where both projective and rotational elements contribute to overall behavior
\end{enumerate}

Traditional econometric approaches analyze these components sequentially—first estimating correlations (projective relationships), then investigating temporal patterns (rotational relationships) through separate procedures. This sequential treatment assumes stable correlation structures during phase relationship analysis, an assumption that fails precisely when both magnitude and timing relationships evolve simultaneously during crisis periods.

\subsection{Multivector Representation of Economic States}

We work in the Clifford algebra $\mathcal{G}(4,0)$ over $\mathbb{R}^4$ with orthonormal basis $\{e_1, e_2, e_3, e_4\}$. Four standardized macroeconomic variables at time $t$ are embedded as:
\begin{equation}
X_t = u_t e_1 + s_t e_2 + r_t e_3 + v_t e_4
\label{eq:vector_embedding}
\end{equation}
where $u_t$, $s_t$, $r_t$, and $v_t$ represent unemployment rate, personal savings rate, consumption growth, and revolving credit growth, respectively.

To capture both individual variable dynamics and their interactions, we construct multivector representations that include scalar, vector, and bivector components:
\begin{equation}
M_t = \alpha_0 + \sum_{i=1}^4 \alpha_i e_i + \sum_{(i,j) \in \Pi} \gamma_{ij}(x_{t,i} - x_{t,j})(e_i \wedge e_j)
\label{eq:multivector_embedding}
\end{equation}

The set $\Pi$ contains six interpretable bivector interactions corresponding to all pairwise relationships between our four variables:
\begin{align}
\Pi = \{&(1,2), (1,3), (1,4), (2,3), (2,4), (3,4)\} \\
\text{representing } &\{u \wedge s, u \wedge r, u \wedge v, s \wedge r, s \wedge v, r \wedge v\}
\end{align}

The difference terms $(x_{t,i} - x_{t,j})$ ensure that bivector components activate when variables diverge, capturing the tension that drives rotational dynamics. When variables move in perfect synchrony, bivector contributions vanish, leaving only projective (vector) relationships.

\subsection{Linear Attention as Economic Pattern Recognition}

Given the multivector representation $M_t$, we project it into query, key, and value spaces through learnable transformations:
\begin{align}
Q_t &= \phi(W_Q M_t), &
K_t &= \phi(W_K M_t), &
V_t &= W_V M_t
\label{eq:qkv_projections}
\end{align}

where $W_Q, W_K, W_V \in \mathbb{R}^{d_h \times d_m}$ are parameter matrices, $d_m$ is the multivector dimension, and $d_h$ is the hidden dimension.

The feature map $\phi(\cdot)$ implements a shifted Leaky ReLU activation:
\begin{equation}
\phi(x) = \begin{cases}
x + 1, & x > 0 \\
\alpha x + 1, & x \leq 0
\end{cases}
\label{eq:activation}
\end{equation}

The shift ensures positivity while the leak parameter $\alpha \in (0,1)$ captures asymmetric economic responses. This reflects empirical evidence that economic relationships exhibit different sensitivities during expansions versus contractions, consistent with loss aversion \citep{kahneman1979} and asymmetric policy transmission mechanisms.

\subsubsection{Attention Mechanism Design}

For lookback horizon $L$, we define the historical window $\mathcal{W}_t = \{t-L, \ldots, t-1\}$ and compute sufficient statistics:
\begin{align}
S_t &= \sum_{\tau \in \mathcal{W}_t} K_\tau V_\tau^\top \\
Z_t &= \sum_{\tau \in \mathcal{W}_t} K_\tau
\label{eq:sufficient_stats}
\end{align}

The attended context vector emerges from the linear attention formulation:
\begin{equation}
O_t = \frac{Q_t^\top S_t}{Q_t^\top Z_t + \varepsilon}
\label{eq:attended_context}
\end{equation}

where $\varepsilon > 0$ provides numerical stability.

\textbf{Economic Justification of Dot Product Similarity}: The dot product $Q_t^\top K_\tau$ measures geometric similarity between current economic conditions (query) and historical states (keys). Because our embedding decomposes into scalar, vector, and bivector components, this similarity metric operates simultaneously across:
\begin{itemize}
\item \textbf{Baseline economic levels} (scalar components)
\item \textbf{Individual variable directions} (vector components)  
\item \textbf{Interaction patterns} (bivector components)
\end{itemize}

This provides a much richer similarity measure than simple Euclidean distance between variable levels, enabling the model to identify historical periods with similar underlying economic dynamics rather than just similar variable values.

The resulting attention weights $w_{\tau,t} = \frac{Q_t^\top K_\tau}{Q_t^\top Z_t + \varepsilon}$ represent the relevance of historical period $\tau$ for understanding current period $t$, implementing a form of economic reasoning by analogy that has deep roots in both economic theory and policy practice.

\subsection{Prediction Architecture}

We implement two alternative prediction heads to map the attended context $O_t \in \mathbb{R}^{d_h}$ to scalar predictions $\hat{y}_t$:

\textbf{Linear Head} (interpretability-focused):
\begin{equation}
\hat{y}_t = w_{\text{out}}^\top O_t + b_{\text{out}}
\label{eq:linear_head}
\end{equation}

This preserves complete interpretability by maintaining linear relationships between attended context components and predictions, enabling exact attribution analysis.

\textbf{MLP Head} (performance-focused):
\begin{equation}
\hat{y}_t = f_{\text{MLP}}(O_t) = W_2 \sigma(W_1 O_t + b_1) + b_2
\label{eq:mlp_head}
\end{equation}

The MLP head introduces controlled nonlinearity that can enhance predictive accuracy while preserving the underlying geometric interpretation of the attention mechanism.

\subsection{Time-Varying Coefficient Interpretation}

Our framework can be equivalently expressed as a time-varying coefficient regression, connecting it directly to established econometric literature:
\begin{equation}
y_t = \beta_{0,t} + \sum_{i=1}^4 \beta_{i,t} x_{i,t} + \sum_{(i,j) \in \Pi} \gamma_{ij,t}(x_{i,t} - x_{j,t}) + \varepsilon_t
\label{eq:tvp_regression}
\end{equation}

The coefficients evolve based on geometric similarity to historical states:
\begin{align}
\beta_{i,t} &= \sum_{\tau=t-L}^{t-1} w_{\tau,t} \beta_{i,\tau} \\
\gamma_{ij,t} &= \sum_{\tau=t-L}^{t-1} w_{\tau,t} \gamma_{ij,\tau}
\label{eq:coefficient_evolution}
\end{align}

This formulation reveals three key innovations over traditional time-varying parameter models:

\textbf{Data-Driven Parameter Evolution}: Rather than assuming random walk dynamics or pre-specified regime switches, coefficients adapt based on learned geometric similarity to historical economic configurations.

\textbf{Structured Interaction Effects}: The $\gamma_{ij,t}$ coefficients directly correspond to bivector components, providing principled guidance for which interactions to include based on geometric algebra theory rather than ad-hoc specification.

\textbf{Interpretable Attention Weights}: The weights $w_{\tau,t}$ provide direct insight into which historical periods most strongly influence current parameter values, enabling transparent analysis of model reasoning.

\subsection{Attribution and Interpretability}

The framework provides multiple layers of interpretability through structured attribution measures:

\textbf{Temporal Attribution}: For prediction at time $T$, normalized attention weights identify the most influential historical periods:
\begin{equation}
w_{\tau,T} = \frac{Q_T^\top K_\tau}{\sum_{j=T-L}^{T-1} Q_T^\top K_j}, \quad \sum_{\tau=T-L}^{T-1} w_{\tau,T} = 1
\end{equation}

\textbf{Geometric Attribution}: For each geometric component $B$ (scalar, vector, or bivector block), we measure contribution through controlled occlusion:
\begin{equation}
\Delta^B \hat{y}_T = \hat{y}_T - \hat{y}_T\big|_{Q_T^{(B)} = 0}
\label{eq:geometric_attribution}
\end{equation}

This reveals which geometric relationships (baseline levels, individual variables, or variable interactions) contribute most to current predictions, providing direct economic insight into the mechanisms driving model outputs.

\section{Empirical Motivation and Analytical Protocol}

\subsection{The Inadequacy of Correlation-Based Crisis Analysis}

Traditional economic analysis relies heavily on correlation measures to understand relationships between macroeconomic variables. However, correlation coefficients provide only aggregate measures of co-movement, obscuring the underlying mechanisms that drive different types of economic crises. To illustrate this fundamental limitation, we analyze U.S. macroeconomic data spanning 1985-2024, focusing on the relationship between unemployment and credit defaults across three distinct crisis periods.

\textbf{Crisis Period Comparison:}
\begin{description}
\item[1990-91 Recession ($\rho = 0.74$):] Sequential causation pattern where unemployment rises first due to oil price shocks and monetary tightening, followed by predictable increases in defaults as affected households exhaust financial buffers. Credit markets remain functional, allowing gradual adjustment.

\item[2008 Financial Crisis ($\rho = 0.78$):] Despite similar correlation magnitude, the underlying mechanism involves simultaneous feedback spirals. Unemployment and credit contraction amplify each other in accelerating cycles: job losses trigger defaults, which tighten credit standards, which reduces economic activity, which increases unemployment.

\item[2020 COVID Crisis ($\rho = 0.32$):] Correlation breakdown reflects policy intervention success. Despite unprecedented unemployment spikes, defaults remain suppressed through forbearance programs, direct payments, and enhanced unemployment benefits. Traditional transmission channels are severed by design.
\end{description}

The similar correlation coefficients between the first two crises mask fundamentally different interaction mechanisms, while the third crisis demonstrates how policy can decouple historically stable relationships. Standard econometric approaches that rely on correlation analysis would classify 1990-91 and 2008 as similar phenomena requiring similar policy responses—a potentially dangerous misdiagnosis.

Our geometric framework addresses this limitation by decomposing variable relationships into projective (correlation-like) and rotational (feedback-spiral) components, enabling distinction between sequential causation and simultaneous amplification even when aggregate correlation measures appear similar.

\subsection{Data Construction and Variable Selection}
\label{sec:data}

We construct a quarterly dataset spanning 1980Q1 to 2024Q2 using Federal Reserve Economic Data (FRED), selecting variables that capture the core transmission channels between macroeconomic conditions and consumer credit stress:

\begin{description}
\item[\textbf{UNRATE (Unemployment Rate)}:] Primary measure of labor market stress, serving as both a lagging indicator of economic conditions and a leading indicator of household financial distress. Unemployment directly affects household income and ability to service debt obligations.

\item[\textbf{PCE (Personal Consumption Expenditure)}:] Core measure of aggregate demand representing approximately 70\% of U.S. GDP. Changes in consumption patterns reflect both current economic conditions and forward-looking household expectations about future income and credit availability.

\item[\textbf{PSAVERT (Personal Saving Rate)}:] Captures precautionary behavior and consumption smoothing dynamics. During normal periods, higher saving rates may reflect confidence-driven consumption deferral. During stress periods, sudden increases may signal panic-driven consumption cuts that amplify economic contractions.

\item[\textbf{REVOLSL (Revolving Consumer Credit Outstanding)}:] Measures household leverage and credit utilization patterns. Unlike fixed-term loans, revolving credit provides real-time information about household liquidity management and financial stress.

\item[\textbf{CORCACBS (Charge-off Rate on Consumer Loans)}:] Target variable representing the percentage of consumer loans written off as uncollectible. This forward-looking measure of credit losses provides early warning of systemic stress in the consumer credit system.
\end{description}

\textbf{Data Preprocessing:} All variables undergo rolling 8-quarter standardization to ensure stationarity while preserving time-varying relationship structures. This approach maintains the relative importance of variables across different economic cycles while removing secular trends that could confound the geometric relationship analysis.

\subsection{Model Architecture and Implementation}

Our implementation focuses on interpretability and economic relevance rather than purely predictive performance. The model architecture incorporates several design choices motivated by economic theory and empirical constraints:

\textbf{Lookback Horizon:} We set $L = 8$ quarters (2 years) to capture business cycle dynamics while maintaining sufficient historical context for pattern recognition. This horizon encompasses typical recession durations and allows the model to identify both short-term stress signals and longer-term structural changes.

\textbf{Bivector Interaction Set:} The model focuses on six economically interpretable bivector interactions representing all pairwise relationships between our four input variables:
\begin{align}
e_1 \wedge e_2 &: \text{unemployment-savings rotation (precautionary vs. forced saving)} \\
e_1 \wedge e_3 &: \text{unemployment-consumption rotation (income vs. confidence effects)} \\
e_1 \wedge e_4 &: \text{unemployment-credit rotation (supply vs. demand factors)} \\
e_2 \wedge e_3 &: \text{savings-consumption rotation (substitution effects)} \\
e_2 \wedge e_4 &: \text{savings-credit rotation (buffer vs. leverage dynamics)} \\
e_3 \wedge e_4 &: \text{consumption-credit rotation (spending-financing feedback)}
\end{align}

Each bivector captures distinct economic mechanisms: unemployment-credit rotation measures labor market transmission to financial stress, while consumption-credit rotation captures demand-financing feedback loops central to crisis amplification.

\subsection{Structured Regularization Framework}

Rather than applying uniform L2 penalties, we implement differentiated regularization that reflects the distinct economic roles of different parameter groups:

\begin{equation}
\mathcal{L}_{\text{reg}} = \lambda_{QK} \left( \|W_Q\|_F^2 + \|W_K\|_F^2 \right) + \lambda_V \|W_V\|_F^2
\label{eq:structured_regularization}
\end{equation}

with $\lambda_{QK} = 10^{-3} > \lambda_V = 10^{-4}$.

\textbf{Economic Rationale:} Query and key matrices ($W_Q, W_K$) determine which historical periods are considered similar to current conditions—the core of economic reasoning by analogy. Stronger regularization ensures attention focuses on genuine historical precedents rather than spurious correlations. The value matrix ($W_V$) maps selected historical information to predictions; weaker regularization allows flexibility in capturing time-varying economic signals while preventing overfitting.

\textbf{Theoretical Justification:} This structure supports our theoretical guarantees (Section 5). Propositions 2 and 3 require bounded operator norms for $W_Q$ and $W_K$ to ensure Lipschitz continuity and well-conditioned attention denominators. The regularization scheme directly implements these theoretical requirements.

\subsection{Mixed-Frequency Modeling for Timely Analysis}

While our target variable (charge-off rates) is only available quarterly, the need for timely economic analysis requires more frequent updates. We address this through a mixed-frequency approach that applies quarterly-trained geometric relationships to monthly input data ~\citet{ghysels2007}.

\subsubsection{Theoretical Foundation}

The approach rests on temporal aggregation theory: economic relationships identified at quarterly frequency reflect underlying processes that operate continuously. Monthly fluctuations represent within-quarter dynamics of the same mechanisms that drive quarterly outcomes. The geometric relationships learned from quarterly data—particularly the bivector coefficients capturing interaction patterns—remain stable across temporal frequencies because they reflect deeper economic structures.

Consider the unemployment-credit bivector $u \wedge v$: the feedback mechanism between job losses and credit tightening operates continuously in the economy, not just at quarter-end measurement points. Monthly data provides higher-frequency observations of this same underlying process.

\subsubsection{Implementation Strategy}

For month $m$ within quarter $q$, we generate nowcasts using:
\begin{equation}
\hat{y}_{m,q} = f(X_m; \Theta_q, W_{j,q})
\label{eq:monthly_nowcast}
\end{equation}

where:
\begin{itemize}
\item $X_m$ represents monthly input values (unemployment, consumption, savings, credit)
\item $\Theta_q$ contains quarterly-trained multivector parameters
\item $W_{j,q}$ are quarterly attention weights identifying relevant historical periods
\end{itemize}

The monthly inputs are standardized using rolling 8-month windows (equivalent to quarterly 8-quarter windows) to maintain consistency in scaling while providing more timely updates to economic conditions. The monthly nowcasting capability transforms our framework from a historical analysis tool into a real-time economic monitoring system, providing policymakers and researchers with timely insights into evolving credit cycle dynamics while maintaining the theoretical rigor of the quarterly model.

\section{Theoretical Foundations}

The theoretical analysis of our framework establishes four fundamental properties that ensure mathematical rigor and economic interpretability. These results demonstrate that our approach constitutes a principled mathematical framework rather than an ad-hoc curve-fitting procedure, while providing the theoretical foundations necessary for reliable economic analysis.

\subsection{Geometric Invariance and Economic Interpretability}

Our first result establishes that the model learns fundamental economic relationships rather than artifacts of data representation or scaling choices.

\begin{proposition}[Geometric Invariance] \label{prop:invariance}
Let $\mathcal{G}(4,0)$ be our geometric algebra space and let $\mathcal{I}: \mathcal{G}(4,0) \to \mathcal{G}(4,0)$ be an isometry represented by rotor $U$ such that $\mathcal{I}(A) = UAU^{-1}$ for any multivector $A$. If input data are transformed as $X'_j = \mathcal{I}(X_j)$ and parameters are consistently transformed as $W'_Q = \mathcal{I}(W_Q)$ and $W'_K = \mathcal{I}(W_K)$, then the attention scores remain invariant: $\alpha'_{ti} = \alpha_{ti}$ where $\alpha_{ti} = Q_t \cdot K_i$.
\end{proposition}

\begin{proof}
Under the transformation, the query and key become:
\begin{align}
Q'_t &= \phi(W'_Q M'_t) = \phi(UW_Q U^{-1} \cdot UM_t U^{-1}) \\
&= \phi(U(W_Q M_t)U^{-1}) = U\phi(W_Q M_t)U^{-1} = UQ_t U^{-1}
\end{align}

Similarly, $K'_i = UK_i U^{-1}$. The attention score becomes:
\begin{align}
\alpha'_{ti} &= Q'_t \cdot K'_i = \langle (UQ_t U^{-1})(UK_i U^{-1}) \rangle_0 \\
&= \langle UQ_t K_i U^{-1} \rangle_0 = \langle Q_t K_i \rangle_0 = \alpha_{ti}
\end{align}

where the third equality uses the fact that the scalar part of a multivector is invariant under conjugation by rotors. 
\end{proof}

\textbf{Economic Interpretation:} This invariance ensures that discovered economic relationships reflect genuine structural patterns rather than arbitrary coordinate system choices or measurement unit effects. For instance, the model's identification of feedback spirals during the 2008 crisis would remain valid regardless of whether unemployment is measured as a percentage or proportion, or whether the data is mean-centered.

\subsection{Stability and Boundedness Properties}

The next results establish that our model produces stable, well-behaved predictions under reasonable conditions.

\begin{proposition}[Bounded Outputs] \label{prop:boundedness}
Assume $\|x_t\|_2 \leq M$ for all $t$ and $\|W_Q\|_F, \|W_K\|_F, \|W_V\|_F \leq C$. Then for lookback window $L$ and attended context $O_t$:
\begin{equation}
\|O_t\|_2 \leq \frac{LM^2C^2}{\varepsilon}
\end{equation}
Consequently, predictions satisfy $|\hat{y}_t| \leq \|w_{\text{out}}\|_2\|O_t\|_2 + |b_{\text{out}}|$.
\end{proposition}

\begin{proof}
From the linear attention formulation:
\begin{align}
\|S_t\|_F &= \left\|\sum_{\tau \in \mathcal{W}_t} K_\tau V_\tau^T\right\|_F \leq \sum_{\tau \in \mathcal{W}_t} \|K_\tau\|_2\|V_\tau\|_2 \\
&\leq L \max_\tau \|K_\tau\|_2 \max_\tau \|V_\tau\|_2 \leq LMC \cdot MC = LM^2C^2
\end{align}

Similarly, $\|Z_t\|_2 \leq LMC$. Therefore:
\begin{equation}
\|O_t\|_2 = \left\|\frac{Q_t^T S_t}{Q_t^T Z_t + \varepsilon}\right\|_2 \leq \frac{\|Q_t\|_2\|S_t\|_F}{\varepsilon} \leq \frac{MC \cdot LM^2C^2}{\varepsilon} = \frac{LM^3C^3}{\varepsilon}
\end{equation}

The prediction bound follows directly from the linear/MLP head structure. 
\end{proof}

\begin{proposition}[Lipschitz Continuity] \label{prop:lipschitz}
Under the conditions of Proposition~\ref{prop:boundedness}, if inputs $x$ and $x'$ satisfy $\|x - x'\|_\infty \leq \delta$, then predictions satisfy:
\begin{equation}
|\hat{y}(x) - \hat{y}(x')| \leq L_{\text{Lip}} \cdot \delta
\end{equation}
where $L_{\text{Lip}}$ depends polynomially on $(M, C, L, \varepsilon^{-1})$.
\end{proposition}

\begin{proof}[Proof Sketch]
The Lipschitz property follows from the composition of Lipschitz functions. The multivector embedding, linear projections, activation function $\phi$, and rational attention computation are all Lipschitz continuous under our boundedness assumptions. The Lipschitz constants can be bounded explicitly in terms of the problem parameters, yielding the stated polynomial dependence. The complete proof involves careful tracking of constants through each computational step. 
\end{proof}

\textbf{Economic Interpretation:} These stability results ensure that small perturbations in economic data produce proportionally small changes in model predictions. This is crucial for policy analysis, as it guarantees that measurement noise or minor data revisions will not cause dramatic swings in model-based assessments of economic conditions.

\subsection{Regularization Theory Connection}

\begin{remark}[Theoretical Justification for Structured Regularization]
The bounded norm assumptions in Propositions~\ref{prop:boundedness} and~\ref{prop:lipschitz} require $\|W_Q\|_F, \|W_K\|_F \leq C$. Our structured regularization scheme (Section 4) directly implements these theoretical requirements by imposing stronger penalties on query and key matrices:
\begin{equation}
\mathcal{L}_{\text{reg}} = \lambda_{QK}(\|W_Q\|_F^2 + \|W_K\|_F^2) + \lambda_V\|W_V\|_F^2
\end{equation}
with $\lambda_{QK} > \lambda_V$. This ensures bounded operator norms in practice while allowing flexibility in the value mapping $W_V$ that doesn't affect stability guarantees.
\end{remark}

\subsection{Impulse Response Analysis}

Our final theoretical result enables precise counterfactual analysis and stress testing capabilities.

\begin{proposition}[Exact Impulse Response] \label{prop:impulse}
Consider a perturbation $\Delta M_\tau$ at historical time $\tau \in [t-L, t-1]$. The induced change in current prediction is:
\begin{equation}
\Delta \hat{y}_t = w_{\text{out}}^T \left[\frac{Q_t^T(\Delta K_\tau V_\tau^T + K_\tau \Delta V_\tau^T)}{Q_t^T Z_t + \varepsilon} - \frac{Q_t^T \Delta K_\tau}{Q_t^T Z_t + \varepsilon} O_t\right]
\end{equation}
where $\Delta K_\tau = \phi'(W_K M_\tau) W_K \Delta M_\tau$ and $\Delta V_\tau = W_V \Delta M_\tau$.
\end{proposition}

\begin{proof}
The prediction $\hat{y}_t = w_{\text{out}}^T O_t$ where $O_t = \frac{Q_t^T S_t}{Q_t^T Z_t + \varepsilon}$. Under perturbation $\Delta M_\tau$:

\begin{align}
\Delta S_t &= \Delta K_\tau V_\tau^T + K_\tau \Delta V_\tau^T \\
\Delta Z_t &= \Delta K_\tau \\
\Delta O_t &= \frac{Q_t^T \Delta S_t}{Q_t^T Z_t + \varepsilon} - \frac{Q_t^T S_t \cdot Q_t^T \Delta Z_t}{(Q_t^T Z_t + \varepsilon)^2} \\
&= \frac{Q_t^T \Delta S_t}{Q_t^T Z_t + \varepsilon} - \frac{Q_t^T \Delta Z_t}{Q_t^T Z_t + \varepsilon} O_t
\end{align}

Substituting the expressions for $\Delta S_t$ and $\Delta Z_t$ yields the stated result. 
\end{proof}

\textbf{Economic Applications:} This exact impulse response formula enables several powerful analytical capabilities:

\begin{itemize}
\item \textbf{Historical Counterfactuals:} "How would current credit conditions differ if unemployment during the 2008 crisis had been 2 percentage points lower?"

\item \textbf{Stress Testing:} "What would happen to predicted charge-offs if we observe unemployment spike similar to early 2020 levels?"

\item \textbf{Policy Analysis:} "How sensitive are current predictions to potential policy interventions that might alter specific historical relationships?"
\end{itemize}

The linearity of the impulse response in the perturbation magnitude allows for straightforward scaling analysis and combination of multiple scenario effects.

\subsection{Generalization Theory}

\begin{corollary}[Generalization Bound]
Under the Lipschitz continuity assumption (Proposition~\ref{prop:lipschitz}), standard Rademacher complexity arguments yield generalization bounds of the form:
\begin{equation}
\mathbb{E}[L_{\text{test}}] - L_{\text{train}} \leq 2L_{\text{Lip}}\mathcal{R}_n + \sqrt{\frac{\log(2/\delta)}{2n}}
\end{equation}
where $\mathcal{R}_n$ is the Rademacher complexity of the function class and $n$ is the training sample size.
\end{corollary}

This bound provides theoretical justification for the model's ability to generalize to new economic conditions while quantifying the trade-off between model complexity (captured by $L_{\text{Lip}}$) and generalization performance.

\subsection{Economic Implications of Theoretical Results}

The theoretical analysis provides several important guarantees for economic applications:

\begin{enumerate}
\item \textbf{Measurement Robustness:} Geometric invariance ensures that economic insights remain valid across different data preprocessing choices or measurement scales.

\item \textbf{Prediction Stability:} Boundedness and Lipschitz continuity guarantee that small data perturbations cannot cause dramatic model behavior changes, essential for reliable policy analysis.

\item \textbf{Counterfactual Validity:} Exact impulse response formulas enable rigorous stress testing and scenario analysis with quantifiable uncertainty bounds.

\item \textbf{Model Reliability:} Generalization bounds provide theoretical basis for trusting model insights on new economic conditions not present in training data.
\end{enumerate}

These theoretical foundations distinguish our approach from purely empirical machine learning methods by providing mathematical guarantees that support reliable economic inference.

\section{Results}
\label{sec:results}

Our empirical analysis fits the GA--linear--attention model using the \emph{MLP head} (Section~3). Throughout this section the objective is explanatory rather than purely predictive: we use model fit, internal states, and learned parameters to uncover the \emph{geometric dynamics} behind credit charge-off cycles.

\subsection{Historical Fit and Model Performance}
\label{sec:hist-fit}

Figure~\ref{fig:historical-fit} compares the model-implied series with realized charge-off rates. The alignment is strong in tranquil periods and remains robust through regime changes. Mechanistically, at each time $t$, the attention weights select historical precedents that are geometrically similar to the current multivector state, and the learned head maps the attended context to a charge-off estimate.

The scalar component captures a slow-moving baseline; vectors (UNRATE, PSAVERT, PCE, REVOLSL) dominate in normal times; and bivectors (pairwise interactions) surge during crises, reflecting feedback loops. The MLP head adds mild nonlinearity that sharpens turning points while preserving interpretability through the GA decomposition.

\begin{figure}[htbp]
  \centering
  \includegraphics[width=0.8\textwidth]{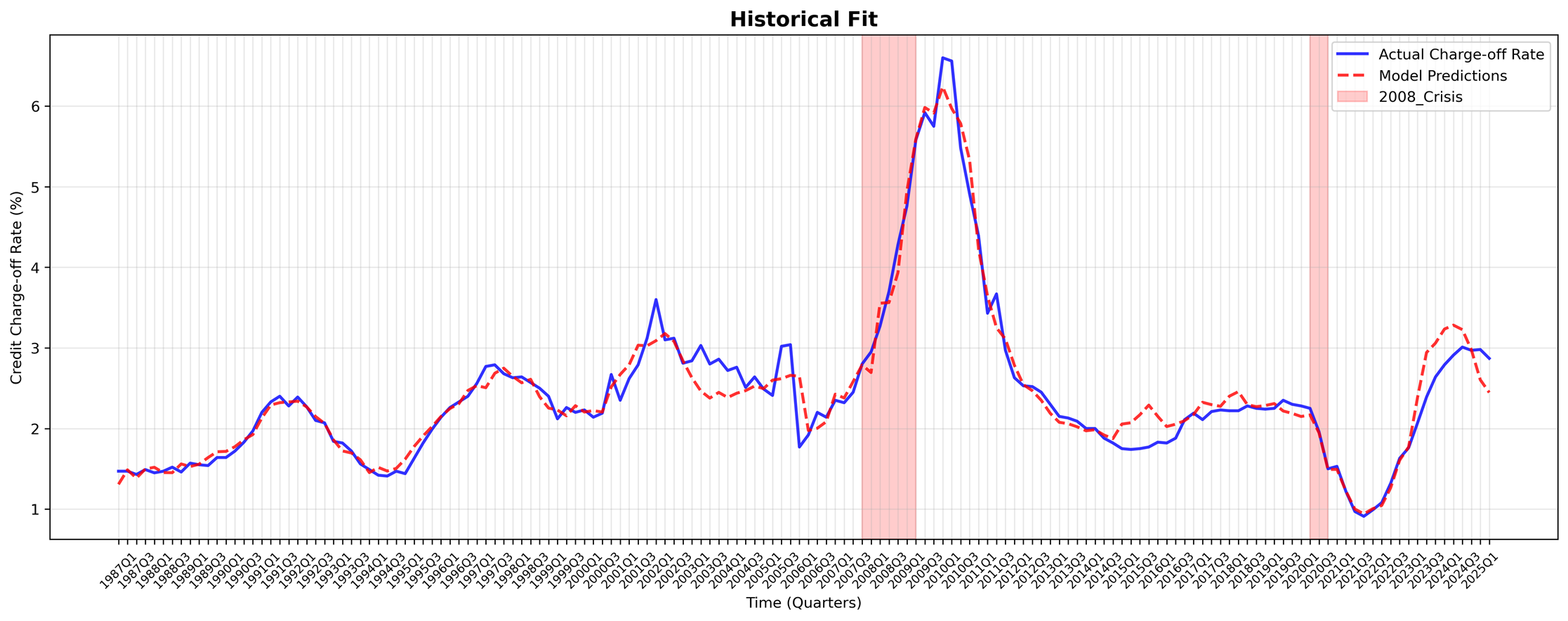}
  \caption{Historical fit using GA embedding, linear attention, and MLP head. Shaded bands mark major crises. Blue line: actual charge-off rate; red dashed line: model predictions.}
  \label{fig:historical-fit}
\end{figure}

\subsection{Geometric Dynamics in Economic State Space}
\label{sec:traj}

Figure~\ref{fig:context-trajectory} shows the trajectory of attended context vectors $\{O_t\}$ in the first two principal components (PC1 explains 50\% of variance; PC2 explains 35.7\%). The path forms loops consistent with business-cycle rotations, superimposed on a secular drift from the 1980s to 2024.

\textbf{Crisis-Specific Geometric Signatures:} The bottom arc (dark points with highest charge-offs) corresponds to the 2008 Financial Crisis; the small cluster near the top (three light points with the lowest charge-offs) corresponds to the 2020 COVID crisis. Despite both being crises, their geometric signatures differ fundamentally: 2008 reflects reinforcing unemployment--credit/consumption spirals; 2020 reflects a sharp reconfiguration driven by policy backstops and altered savings behavior.

\textbf{Economic Hysteresis:} Outside crises, the system traces compact loops with moderate charge-offs. Crucially, return paths differ from entry paths, indicating hysteresis where balance-sheet repair, policy adaptations, and behavioral changes alter the route back to stability. This geometric evidence of path dependence supports theoretical models of economic scarring and institutional learning.

\begin{figure}[htbp]
  \centering
  \includegraphics[width=0.8\textwidth]{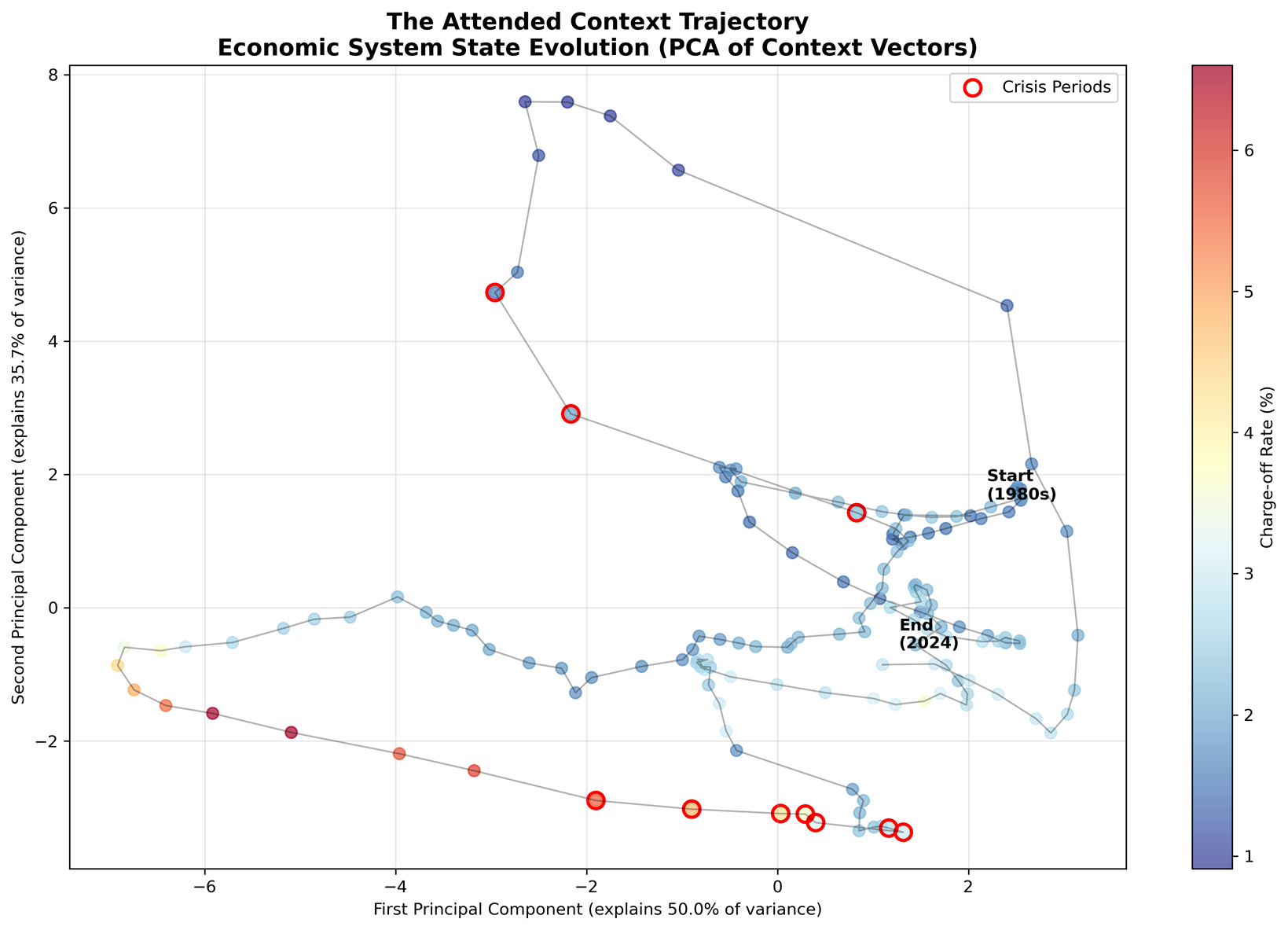}
  \caption{Attended context trajectory (PCA of context vectors). Color intensity denotes charge-off rate; red circles mark crisis quarters. The bottom arc represents the 2008 crisis; the top cluster represents the COVID crisis.}
  \label{fig:context-trajectory}
\end{figure}

\subsection{Temporal Evolution of Geometric Components}
\label{sec:comp-evolution}

Figure~\ref{fig:component-evolution} presents a heatmap (log scale) of scalar, vector, and bivector magnitudes from 1980Q1--2024Q2, revealing how different geometric relationships dominate across economic cycles.

\textbf{Scalar Baseline:} The scalar component remains persistent and stable, acting as a secular anchor rather than a volatility driver. This baseline captures long-term trends in credit conditions independent of cyclical dynamics.

\textbf{Vector Components:} Individual variable effects fluctuate with economic intuition: unemployment surges during recessions, savings rise in downturns, and credit growth/consumption contract around crises. Vectors dominate during normal regimes, reflecting predictable linear relationships between economic indicators and credit outcomes.

\textbf{Bivector Components:} Interaction effects spike dramatically during crises, but with distinct patterns. In 2008, unemployment--credit and unemployment--consumption bivectors dominate, reflecting the feedback spirals between job losses and credit market dysfunction. In 2020, savings-related bivectors rise sharply due to policy-driven savings accumulation with constrained spending opportunities. This demonstrates that crises differ not only in magnitude but crucially in \emph{which geometric relationships} become active.

\begin{figure}[htbp]
  \centering
  \includegraphics[width=0.9\textwidth]{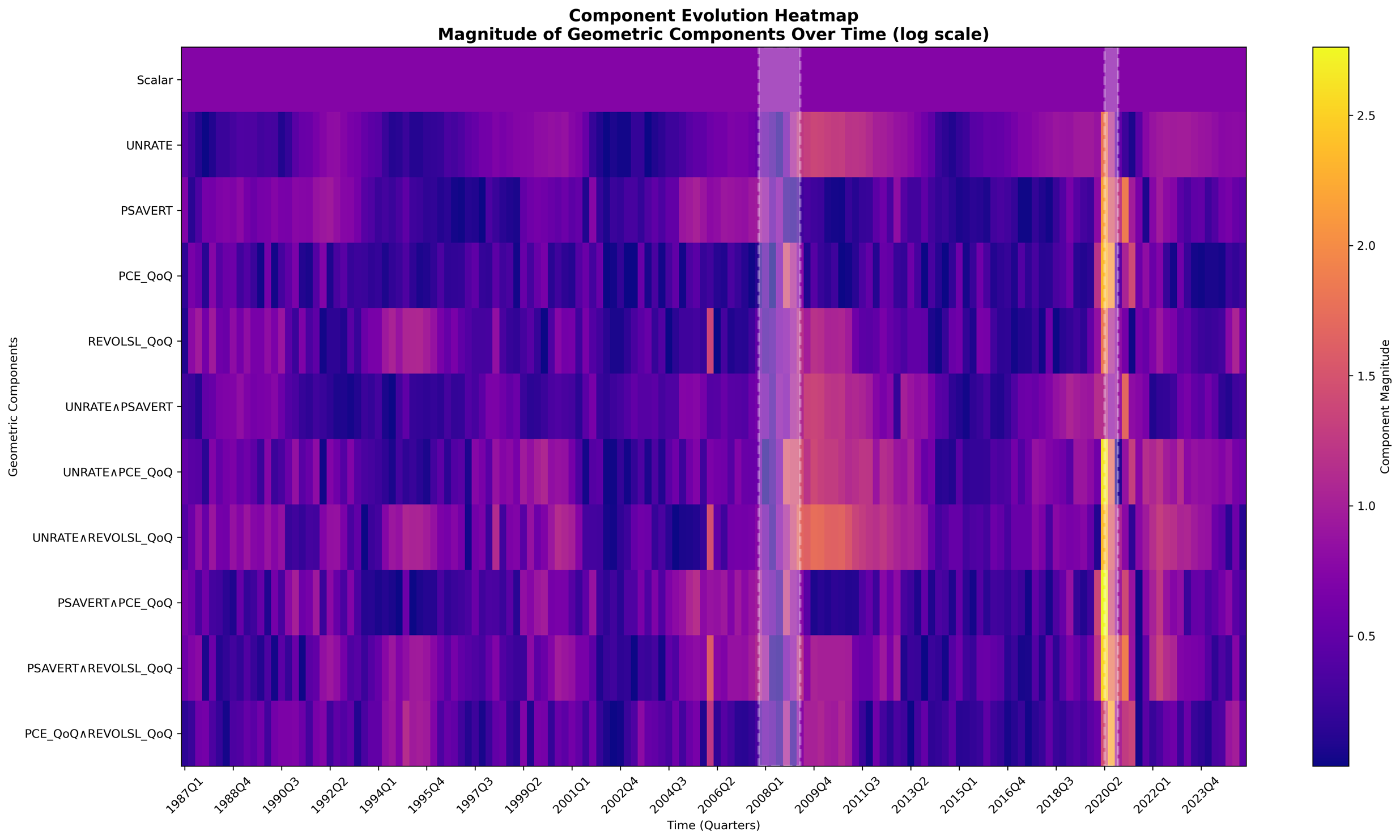}
  \caption{Component evolution heatmap (log scale). Rows represent scalar, vector, and bivector components; columns represent quarters from 1980Q1--2024Q2. Crisis periods exhibit sharp bivector spikes with different interaction structures: unemployment--credit/consumption dominance in 2008; savings-related bivectors in 2020.}
  \label{fig:component-evolution}
\end{figure}

\subsection{Attention Patterns and Historical Analogies}
\label{sec:attn-heatmap}

Figure~\ref{fig:attention-heatmap} visualizes attention weights across the sample period. Each column represents a current quarter; each row represents a lookback lag from $t-1$ to $t-8$.

\textbf{Recency Bias in Normal Times:} Relative to perfect recency weighting, the model shows recency dominance during stable periods but allocates significant weight to deeper lags when current conditions do not closely resemble the immediate past.

\textbf{Crisis-Period Pattern Recognition:} Around 2008 and 2020, attention spreads broadly across lags $t-4$ to $t-8$, indicating the model searches for deeper historical analogies rather than relying solely on recent precedents. This behavior reflects the model's recognition that crisis periods often require reference to much earlier historical experiences.

\begin{figure}[htbp]
  \centering
  \includegraphics[width=0.9\textwidth]{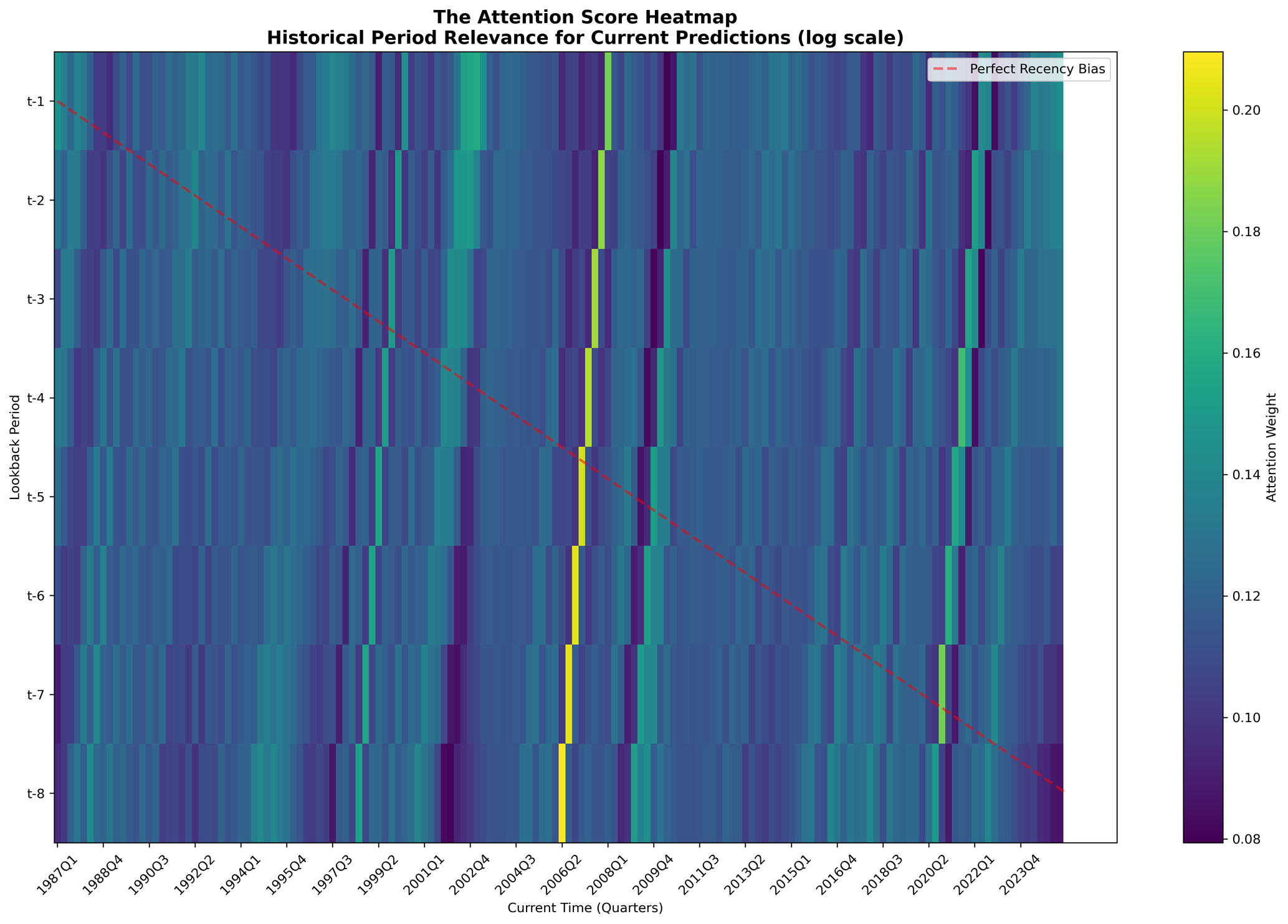}
  \caption{Attention score heatmap (log scale). Columns: current time; rows: lookback lags ($t-1$ to $t-8$). The dashed line marks perfect recency weighting. Crisis periods show broader weight dispersion across lags; normal times exhibit stronger recency focus.}
  \label{fig:attention-heatmap}
\end{figure}

\subsection{Detailed Attention Analysis for Key Periods}
\label{sec:attn-panels}

Figure~\ref{fig:attention-distribution} examines attention distributions for representative quarters, revealing how the model's historical analogies evolve across different economic conditions.

\textbf{Normal Periods (1987Q1--Q2):} Attention weights remain relatively uniform across lags, consistent with stable economic dynamics where recent and distant history provide similar informational value.

\textbf{Crisis Onset (2007Q4):} Attention concentrates sharply on lag $t-2$, reflecting the model's recognition of short-horizon stress signals as leading indicators of emerging systemic problems.

\textbf{Crisis Peak (2009Q3--2010Q2):} Weights spread widely across lags $t-4$ through $t-8$, consistent with complex feedback dynamics requiring reference to deeper historical precedents for pattern recognition.

\begin{figure}[htbp]
  \centering
  \includegraphics[width=0.8\textwidth]{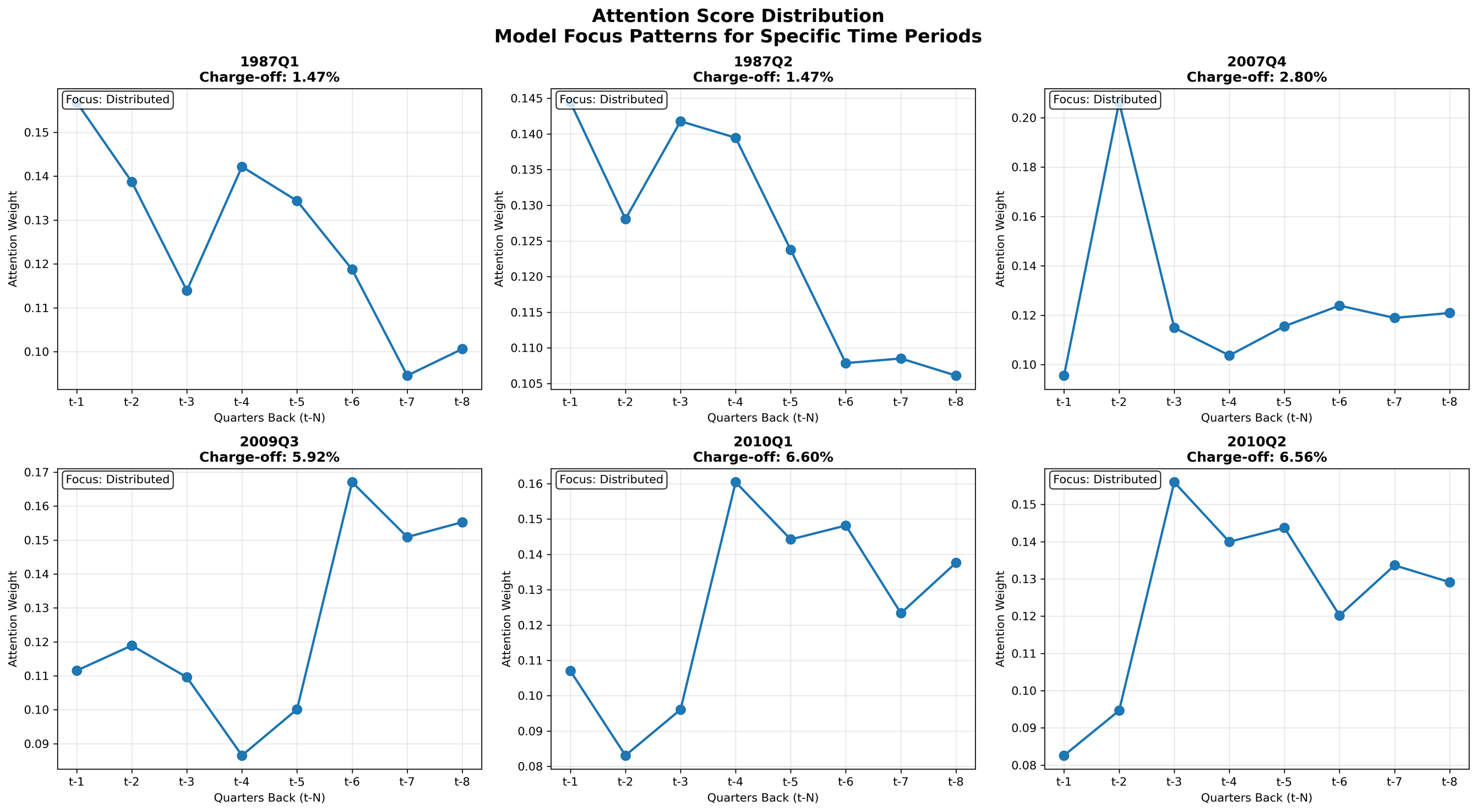}
  \caption{Attention score distributions for selected quarters. Normal periods show balanced attention across lags; crisis onset exhibits sharp short-horizon focus; crisis peaks demonstrate broad historical search patterns.}
  \label{fig:attention-distribution}
\end{figure}

\subsection{Parameter Structure and Economic Interpretation}
\label{sec:param-structure}

Figure~\ref{fig:parameter-magnitudes} summarizes average component magnitudes across the query, key, and value projection matrices ($W_Q$, $W_K$, and $W_V$), revealing how different geometric relationships contribute to model reasoning.

\textbf{Query and Key Matrices ($W_Q$, $W_K$):} Vector and bivector components show substantial magnitudes (with scalars relatively smaller in $W_Q$ and larger in $W_K$), consistent with their role in defining similarity kernels where interaction geometry plays a central role in historical pattern matching.

\textbf{Value Matrix ($W_V$):} Scalar magnitude dominates, indicating that once relevant historical periods are identified through geometric similarity, baseline stress levels transmit most strongly to current predictions. Vector and bivector components remain meaningful for carrying specific signals and interactions forward.

The observed pattern reflects our structured regularization approach: tighter control on $W_Q$ and $W_K$ stabilizes the attention mechanism's similarity computations, while looser regularization on $W_V$ allows flexible information transmission from identified historical precedents.

\begin{figure}[htbp]
  \centering
  \includegraphics[width=0.7\textwidth]{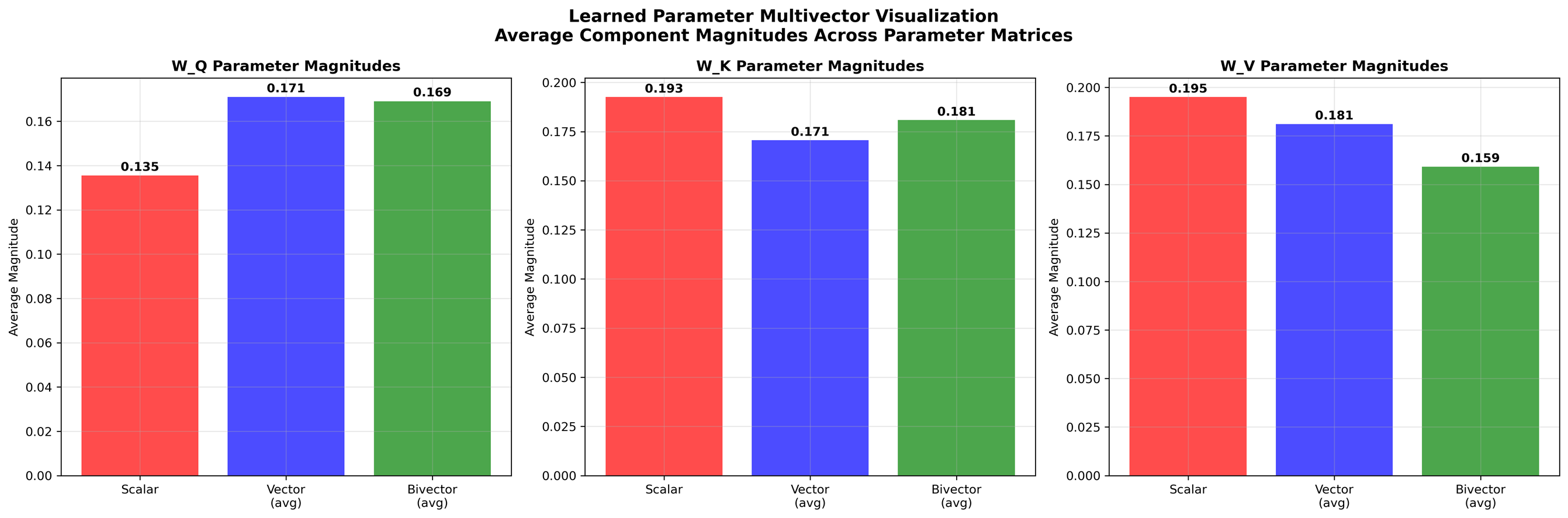}
  \caption{Average component magnitudes across parameter matrices. Query and key matrices ($W_Q$, $W_K$) emphasize vector and bivector components for robust geometric similarity; value matrix ($W_V$) emphasizes scalar components for baseline stress transmission.}
  \label{fig:parameter-magnitudes}
\end{figure}

\subsection{Variable-Specific Attribution Analysis}
\label{sec:var-contr}

Figure~\ref{fig:variable-contributions} decomposes the attended context into contributions from individual economic variables (UNRATE, PSAVERT, PCE, REVOLSL), revealing how different variables drive model predictions across economic cycles.

\textbf{Normal Times:} REVOLSL (credit utilization) and PSAVERT (savings behavior) provide stable explanatory weight, with moderate and balanced contributions across all variables reflecting steady-state economic relationships.

\textbf{2008 Crisis:} All variables surge dramatically, with unemployment (UNRATE) and credit (REVOLSL) showing particularly strong increases. This broad-based surge reflects the systemic nature of the 2008 crisis, consistent with the dominance of unemployment--credit/consumption bivectors observed in the component analysis.

\textbf{2020 COVID Crisis:} PSAVERT dominates overwhelmingly due to policy-driven savings accumulation, while REVOLSL and PCE contributions contract sharply. Unemployment rises substantially but defaults remain suppressed due to policy backstops, creating an unprecedented decoupling of traditional relationships.

\begin{figure}[htbp]
  \centering
  \includegraphics[scale=1.0]{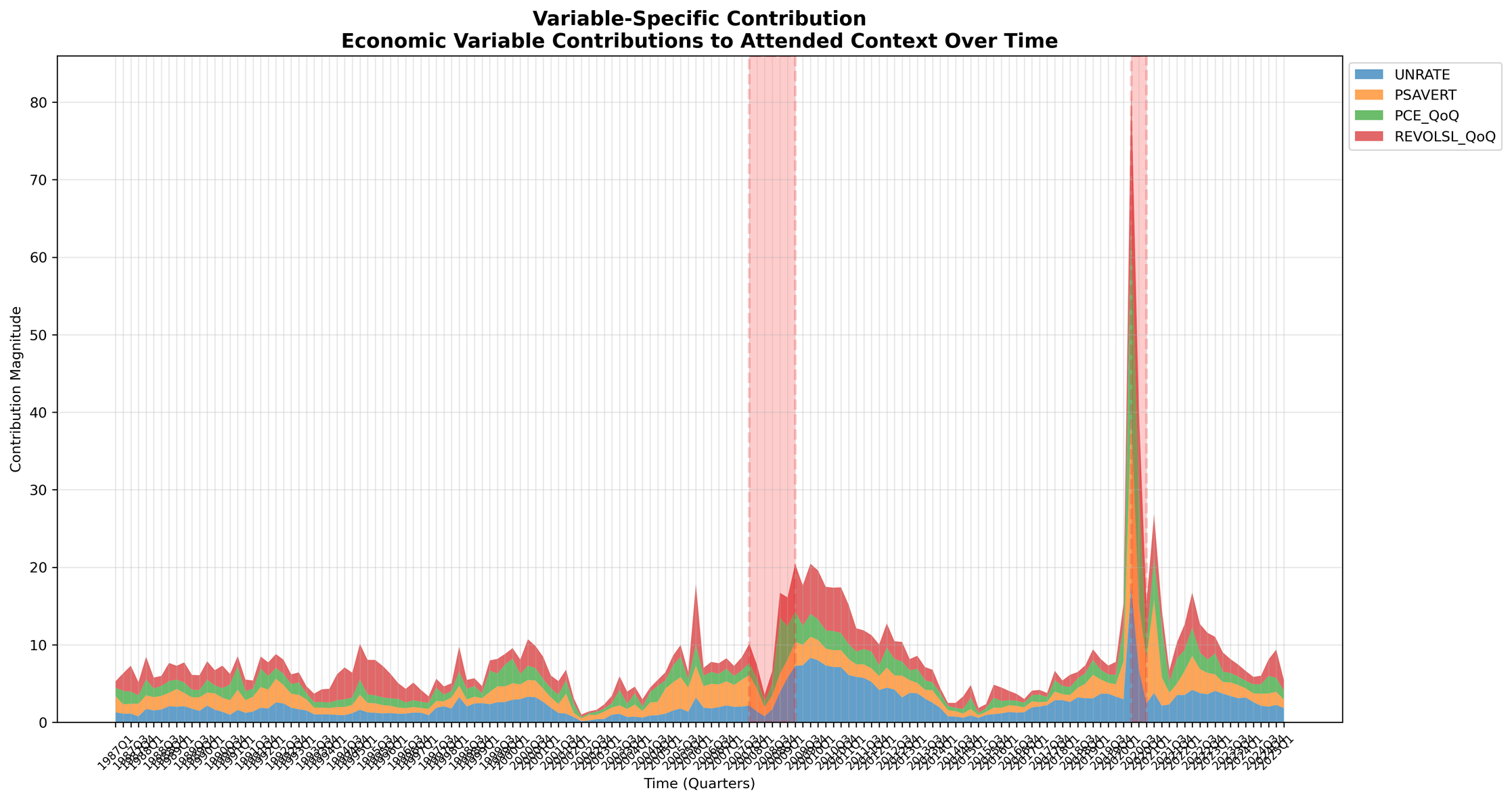}
  \caption{Variable-specific contributions to attended context (stacked areas). The 2008 crisis shows broad surges across all variables; the 2020 crisis exhibits a distinct PSAVERT spike with contracted contributions from credit and consumption variables.}
  \label{fig:variable-contributions}
\end{figure}

\subsection{Contemporary Economic Analysis (2021--2024)}
\label{sec:current}

The model's analysis of recent economic conditions provides insights into the current state of credit cycle dynamics and their historical parallels.

\textbf{Current Economic Positioning:} Post-COVID through 2024, charge-offs exhibit smaller peaks relative to 2008 levels but remain economically meaningful. The trajectory plot shows recent points moving toward ``normal loops'' while remaining offset from pre-2019 patterns, indicating that the post-COVID economy has not returned to its previous state but has stabilized in a new configuration.

Component analysis reveals mild elevations in bivector activity during 2022--2023, particularly in unemployment--savings and consumption--credit interactions. The attention heatmap shows broader-than-normal allocation across historical lags. Both patterns are weaker than 2008 crisis levels but stronger than the tranquil mid-2010s period.

\textbf{Historical Analogies:} Current conditions most closely resemble the 1990--91 recession and 2001--02 credit downturn. In these episodes, vector effects (individual variable impacts) dominated while bivector interactions remained contained. The peaks observed are business-cycle aligned rather than system-threatening, with feedback loops not achieving the dominance characteristic of systemic crises.

\textbf{Underlying Mechanisms:} Relative to 2008, contemporary stress appears vector-dominated with limited geometric amplification. The COVID era demonstrated that policy backstops could break traditional unemployment-to-defaults transmission channels. The lingering effects remain visible as slightly elevated but controlled bivector activity, suggesting the economic geometry favors additive rather than rotational feedback dynamics.

\subsection{Practical Implications and Policy Insights}
\label{sec:lessons}

The geometric analysis yields several actionable insights for economic monitoring and policy formulation:

\begin{enumerate}
\item \textbf{Crisis Differentiation through Geometry:} The distinction between manageable economic stress and systemic crisis is fundamentally geometric rather than scalar. Vector dominance signals conventional cyclical stress amenable to standard policy tools; bivector dominance indicates dangerous feedback loops requiring more aggressive intervention.

\item \textbf{Policy Effectiveness Assessment:} The COVID experience demonstrates that well-designed policy interventions can alter geometric relationships, breaking traditional transmission channels. The geometric framework provides real-time assessment of whether policy measures are succeeding in containing feedback loops.

\item \textbf{Economic Hysteresis Recognition:} Recovery paths systematically differ from crisis entry paths, reflecting permanent changes in economic relationships. This geometric evidence of hysteresis has important implications for policy design and expectations about return to pre-crisis conditions.

\item \textbf{Early Warning System Design:} Monitoring specific bivector combinations (unemployment$\wedge$credit, consumption$\wedge$credit) and attention dispersion patterns provides more informative signals than traditional level-based indicators. The transition from vector to bivector dominance serves as a key early warning signal.
\end{enumerate}

\begin{table}[htbp]
\centering
\caption{Economic regime signatures identified by the geometric algebra framework.}
\label{tab:regime-summary}
\small
\begin{tabular}{@{}p{2.2cm}p{2cm}p{2.8cm}p{2.5cm}p{3.5cm}@{}}
\toprule
Economic Regime & Charge-off Level & Dominant Geometry & Attention Pattern & Underlying Mechanism \\
\midrule
Normal (mid-2010s) & Low & Vector dominance & Narrow (recency-focused) & Additive relationships, stable transmission \\
1990-91 recession & Moderate & Vectors $>$ bivectors & Mildly broader & Cyclical stress, contained feedback \\
2001-02 downturn & Moderate & Vectors $>$ bivectors & Mildly broader & Sector-specific slowdown \\
2008 crisis & Very high & Bivectors $\gg$ vectors & Broad across lags & Systemic feedback spirals \\
2020 COVID & Low (despite unemployment) & Savings-related bivectors & Broad but temporary & Policy-disrupted transmission \\
2022-24 recent & Moderate peaks & Vectors with mild bivectors & Broader than normal & Cyclical stress with buffers \\
\bottomrule
\end{tabular}
\end{table}

\textbf{Assessment of Current Risk:} Recent economic conditions exhibit geometric signatures more consistent with manageable cyclical episodes (1990--91, 2001--02) than with systemic crisis (2008). While continued vigilance remains appropriate, the key risk indicator to monitor is a geometric shift toward sustained bivector dominance coupled with persistent attention dispersion across deeper historical lags. Such a transition would signal the emergence of dangerous feedback dynamics requiring immediate policy attention.

\section{Discussion and Conclusion}

\subsection{Principal Findings and Economic Insights}

This study demonstrates that the geometric structure of variable interactions provides critical information about economic dynamics that traditional correlation-based approaches systematically miss. Our geometric algebra framework with linear attention reveals three fundamental insights about credit cycle dynamics:

\textbf{Crisis Mechanisms Are Geometrically Distinct.} The 2008 financial crisis and 1990-91 recession exhibited nearly identical unemployment-default correlations ($\rho \approx 0.78$ vs. 0.74), yet our framework reveals fundamentally different interaction patterns. The 2008 crisis was characterized by dominant bivector components reflecting simultaneous feedback spirals between unemployment, credit contraction, and consumption collapse. In contrast, the 1990-91 recession showed vector-dominated dynamics with predictable sequential relationships. This geometric distinction has profound implications: correlation-based early warning systems would classify these crises as similar, potentially leading to inappropriate policy responses.

\textbf{Policy Interventions Alter Geometric Signatures.} The COVID-19 crisis provides a natural experiment in how policy backstops affect economic transmission mechanisms. Despite severe unemployment shocks, our model identified strong savings-related bivector activity coupled with suppressed default rates---a geometric signature indicating that traditional unemployment-to-defaults transmission channels were effectively severed by policy intervention. This finding suggests that geometric monitoring could help policymakers assess the effectiveness of crisis interventions in real-time.

\textbf{Economic Hysteresis Has Geometric Manifestations.} The attended context trajectory (Figure 2) reveals that economic systems do not retrace their paths when recovering from stress. Instead, they exhibit geometric hysteresis---the return path differs systematically from the entry path into crisis. This finding provides formal evidence for balance sheet repair effects, institutional learning, and behavioral adaptation that alter economic relationships permanently.

\subsection{Methodological Contributions}

Our framework addresses three fundamental limitations in macroeconomic time series analysis:

\textbf{Sequential Treatment of Magnitude and Phase.} Traditional econometric approaches first estimate correlation structures, then separately analyze temporal patterns through Granger causality or impulse responses. This assumes stable correlations while phase relationships evolve---an assumption that fails precisely during crisis periods when both magnitude and timing relationships change simultaneously. Our geometric product $ab = a \cdot b + a \wedge b$ captures both projection and rotational relationships in a single mathematical object, enabling unified analysis of magnitude and phase dynamics.

\textbf{Static Parameter Assumptions.} While regime-switching models allow for discrete parameter changes, they require pre-specification of regimes and switching mechanisms. Our attention-based approach endogenously identifies which historical periods are most relevant for current conditions, providing a data-driven alternative to structural break models. The model learns that 2022-2023 conditions most closely resemble the 1990-91 and 2001-02 periods rather than 2008, without prior specification of these regime classifications.

\textbf{Limited Interaction Modeling.} Standard VAR models capture linear relationships between variables but struggle with complex, time-varying interaction effects. Our bivector components directly measure rotational relationships that distinguish feedback spirals from simple lead-lag dynamics. This capability is crucial for understanding why some economic shocks amplify into systemic crises while others remain contained.

\subsection{Practical Applications and Policy Implications}

The framework provides actionable insights for three key stakeholder groups:

\textbf{Financial Risk Management.} Banks and financial institutions can monitor bivector activity as an early warning system for credit losses. A transition from vector-dominated (additive stress) to bivector-dominated (feedback spiral) dynamics signals heightened systemic risk requiring more aggressive provisioning and capital management. The current analysis suggests 2022-2024 conditions resemble manageable cyclical stress rather than systemic crisis, but vigilance for sustained bivector elevation remains warranted.

\textbf{Monetary Policy.} Central banks can use geometric signatures to distinguish between different types of economic stress requiring different policy responses. Vector-dominated stress may respond well to standard interest rate adjustments, while bivector-dominated dynamics may require unconventional interventions to break feedback loops. The attention mechanism also provides real-time assessment of how current conditions compare to historical precedents, informing policy calibration.

\textbf{Macroeconomic Research.} The framework offers new tools for testing economic theories about crisis transmission, policy effectiveness, and structural change. Researchers can examine whether specific theoretical mechanisms (such as bank lending channels or wealth effects) generate predicted geometric signatures in the data.

\subsection{Limitations and Future Research Directions}

Several limitations warrant acknowledgment and suggest productive avenues for future research:

\textbf{Sample Size Constraints.} Crisis periods are inherently rare events, limiting the statistical power for identifying robust patterns in extreme regimes. While our framework provides interpretable insights about historical crises, confidence intervals around crisis-period parameters remain necessarily wide. Future work could address this through simulation studies or applications to higher-frequency financial data where ``crisis-like'' conditions occur more frequently.

\textbf{Geometric Interpretation Assumptions.} Our economic interpretation of bivector components as feedback mechanisms rests on the assumption that rotational relationships in geometric algebra correspond to meaningful economic dynamics. While this interpretation proves empirically useful and theoretically motivated, alternative geometric frameworks (such as differential geometry or topological methods) might provide different insights into economic relationships.

\textbf{Computational Scalability.} Although linear attention reduces computational complexity compared to standard attention mechanisms, the geometric algebra operations and multivector representations add overhead relative to simple regression approaches. Future implementations could explore approximate methods or specialized hardware acceleration for larger-scale applications.

\textbf{Variable Selection and Embedding Choices.} Our choice of four macroeconomic variables and specific bivector interactions reflects economic intuition but remains somewhat arbitrary. Systematic approaches to variable selection in geometric algebra spaces, perhaps using information-theoretic criteria or automated model selection, could enhance the framework's robustness and generalizability.

\subsection{Broader Implications for Economic Analysis}

This study contributes to a growing literature that applies modern machine learning techniques to understand economic relationships rather than simply improve forecasting accuracy. Several broader implications emerge:

\textbf{Beyond Prediction to Explanation.} The attention mechanism provides direct insight into which historical periods inform current predictions, transforming the model from a black box into an interpretable analytical tool. This capability addresses persistent criticism of machine learning approaches in economics that prioritize predictive performance over economic understanding.

\textbf{Geometric Perspective on Economic Relationships.} By representing economic variables as multivectors rather than scalar time series, we access information about variable interactions that scalar analysis cannot capture. This geometric perspective may prove valuable for other areas of economics where complex, time-varying relationships are central---such as international trade, labor market dynamics, or asset pricing.

\textbf{Integration of Traditional and Modern Methods.} Rather than replacing classical econometric approaches, our framework provides complementary insights. The time-varying coefficient interpretation (Section 3.4) shows how geometric algebra models relate to established econometric theory while adding capabilities for analyzing interaction dynamics.

\subsection{Conclusion}

The geometric dynamics of consumer credit cycles reveal systematic patterns that traditional correlation-based analysis obscures. By representing economic relationships as multivectors and using attention mechanisms to identify relevant historical precedents, we uncover the rotational structure of variable interactions that distinguishes manageable cyclical stress from systemic crisis.

The framework's key innovation lies not in superior predictive performance, but in providing interpretable insights into the mechanisms driving economic dynamics. Bivector components capture feedback relationships that correlation cannot detect, while attention weights reveal which historical periods are most relevant for understanding current conditions. Together, these components offer a new lens for analyzing economic relationships that captures both their magnitude and geometric structure.

Current economic conditions (2022-2024) exhibit geometric signatures most similar to the 1990-91 recession and 2001-02 downturn rather than the 2008 financial crisis. While charge-off rates show moderate elevation, the underlying dynamics remain vector-dominated with contained bivector activity, suggesting manageable cyclical stress rather than systemic fragility. However, continued monitoring for a transition toward sustained bivector dominance and persistent attention dispersion remains essential, as these geometric shifts would signal the onset of more dangerous feedback dynamics.

The broader contribution extends beyond credit cycle analysis to demonstrate how geometric algebra provides a mathematically principled framework for understanding the complex, time-varying relationships that characterize economic systems. By simultaneously capturing both the magnitude and geometric structure of variable interactions, this approach offers new tools for explanatory macroeconomic analysis that complement traditional econometric methods while addressing their fundamental limitations in analyzing dynamic, non-linear economic relationships.

\bibliographystyle{abbrvnat}

\end{document}